\documentclass[letterpaper, 10 pt, conference]{ieeeconf}

\IEEEoverridecommandlockouts      

\usepackage{amsthm}
\usepackage{color}

\usepackage{cite}
\usepackage{amsmath,amssymb,amsfonts}
\usepackage{algorithm}
\usepackage{algorithmic}
\usepackage{graphicx}
\usepackage{subcaption}
\usepackage{textcomp}
\usepackage{array}
\usepackage{booktabs}
\usepackage[para,flushleft]{threeparttable}
\usepackage{lipsum}
\allowdisplaybreaks[4]

\usepackage{cite}
\usepackage{bm}
\usepackage{footmisc}
\usepackage{geometry}
\DeclareCaptionLabelFormat{custom}{#1(#2)}

\newtheorem{lem}{Lemma}
\newtheorem{thm}{Theorem}

\newtheorem{assum}{Assumption}
\newtheorem{definition}{Definition}

\newtheorem{rem}{Remark}
\newtheorem{cor}{Corollary}

\usepackage{lipsum}
\usepackage{booktabs}



\title{\LARGE \bf
	Compression-based Privacy Preservation for Distributed Nash Equilibrium Seeking in Aggregative Games}
\author{Wei Huo, Xiaomeng Chen, Kemi Ding, Subhrakanti Dey, and Ling Shi
\thanks{W. Huo, X. Chen, and L. Shi are with the Department of Electronic and Computer Engineering, Hong Kong University of Science and Technology, Clear Water Bay, Kowloon, Hong Kong (email: whuoaa@connect.ust.hk, xchendu@connect.ust.hk, eesling@ust.hk).}
\thanks{K. Ding is with the School of System Design and Intelligent Manufacturing, Southern University of Science and Technology, Shenzhen, 518055, China (email: dingkm@sustech.edu.cn)}
\thanks{S. Dey is with the Department of Electrical Engineering, Uppsala
University, 75103 Uppsala, Sweden (email: subhrakanti.dey@angstrom.uu.se).}}
\begin{document}
\maketitle
\thispagestyle{empty}
\pagestyle{empty}

\begin{abstract}
 This paper explores distributed aggregative games in multi-agent systems. Current methods for finding distributed Nash equilibrium require players to send original messages to their neighbors, leading to communication burden and privacy issues. To jointly address these issues, we propose an algorithm that uses stochastic compression to save communication resources and conceal information through random errors induced by compression. Our theoretical analysis shows that the algorithm guarantees convergence accuracy, even with aggressive compression errors used to protect privacy. We prove that the algorithm achieves differential privacy through a stochastic quantization scheme. Simulation results for energy consumption games support the effectiveness of our approach.
\end{abstract}


\section{Introduction} \label{sec: introduction}
Aggregative games, which model competitive interactions among players, have seen a surge in interest for applications like network resource allocation~\cite{salehisadaghiani2016distributed} and energy management~\cite{wang2021distributed}. 
In these games, each player's objective function is influenced by all players' strategies, with a Nash equilibrium (NE) characterizing a stable solution where no player intends to unilaterally change its decision.
In decentralized networks lacking a central coordinator, despite players' competitive interests in games, they require specific communication protocols to share information with neighbors to address the absence of global information.

Despite progress in distributed NE seeking (DNES)~\cite{ye2017distributed, salehisadaghiani2016distributed,huang2022linearly, zhu2021networked}, 
privacy concerns and communication burden arise from traditional message broadcasting approaches. Directly transmitting sensitive data, such as power consumption patterns in energy management games, can compromise user privacy and security. Furthermore, communication bandwidth and power are always limited in practical distributed networks.
Thus, it is vital to develop privacy-preserving and communication-efficient algorithms that ensure convergence to NE.

Ensuring privacy in decentralized networks is a complex task. Encryption is commonly used but incurs significant computational overhead~\cite{lu2018privacy}.
The other approach is adding perturbation to achieve differential privacy (DP).
Ye et al.~\cite{ye2021differentially} and Lin et al.~\cite{lin2023statistical} utilized noise to obscures local aggregate estimates to preserve DP.
Chen and Shi~\cite{chen2023differentially} ensured DP by perturbing players' payoff functions using stochastic linear-quadratic functional perturbation.  
However, these methods introduce a trade-off between privacy and convergence accuracy in DNES.
Recently, Wang et al.~\cite{wang2022differentially} perturbed the gradient instead of transmitting data to guarantee almost sure convergence to NE and achieve DP per iteration.

Although the above works have explored privacy-preserving DNES algorithms, they require substantial data transmission during iterative communication with neighbors.
While some works have employed event-triggered mechanisms to reduce communication rounds~\cite{huo2024distributed}, players still transmit original messages if a certain event is triggered. 
Other studies have used compression techniques to reduce transmitted data size in games, including deterministic quantizations~\cite{chen2022distributed}, adaptive quantizations~\cite{pei2023distributed}, and general compressors~\cite{chen2022linear}. 
However, these approaches did not explicitly consider privacy preservation, and quantifying the privacy level arising from compression remains challenging.
Recently, Wang and Ba{\c{s}}ar~\cite{wang2022quantization} demonstrated that the quantization can be leveraged to guarantee DP for distributed optimization, inspiring for utilizing the inherent randomness of stochastic compression to achieve DP and reduce communication costs simultaneously.
Specifically,~\cite{wang2022quantization} combined the consensus and gradient descent to reach an agreement on the optimal solution in distributed optimization.
Nonetheless, the coupling among players' objective functions in aggregative games makes the convergence analysis more difficult.


Motivated by the above observations, this paper jointly considers privacy issues and communication efficiency in DNES.
Unlike previous works that handle these aspects in a cascade fashion~\cite{xie2023compressed},
we propose a novel DNES algorithm that directly utilizes the intrinsic randomness from stochastic compression to protect privacy. 
To ensure a strong privacy guarantee, the bound of the compression error variance does not vanish in our algorithm, which brings challenges for algorithm design and analysis. Without appropriate treatment for the compression errors, the algorithm will diverge due to the error accumulation.
Some works employ the dynamic scaling compression technique to tackle this challenge~\cite{yi2022communication, liao2023linearly}. Directly using this technique will cause exponential growth of the privacy loss per iteration and thus lose privacy.
Thus, instead of dynamically scaling the compressed value, we dedicatedly design the step sizes to reduce the effect of the non-vanishing compression errors and ensure convergence accuracy. Due to the space limitation, the detailed proof of this paper is provided in~\cite{appendix}.
In Table~\ref{tab: compare}, we compare our work with some related works.

\begin{table}[htbp] 
    \centering
    \footnotesize
    \captionsetup{font=small}
    \caption{\small{Comparison with some related works}}
    \begin{threeparttable}
    \begin{tabular}[c c c c]{|m{0.15\linewidth}<{\centering}|m{0.2\linewidth}<{\centering}|m{0.18\linewidth}<{\centering} |m{0.2\linewidth}<{\centering}|}
        \hline
        Work\tnote{$\sharp$} & Communication Reduction & Privacy Preservation & Convergence Accuracy\tnote{$\star$} \\
        \hline
        \cite{ye2021differentially, lin2023statistical, chen2022distributed} &  $\times$ & $(\varepsilon, 0)$-DP & $\omega$-NE \\
        \hline
        \cite{yi2022communication, chen2022linear, liao2023linearly} & $\checkmark$ & $\times$  & Exact NE \\
        \hline
        \cite{xie2023compressed} & $\checkmark$ & $(\varepsilon, 0)$-DP & $\omega$-NE \\
        \hline
        Our work & $\checkmark$ & $(0, \delta)$-DP & Exact NE \\
        \hline
    \end{tabular}
    \begin{tablenotes}
     \item[$\sharp$] \cite{yi2022communication, liao2023linearly, xie2023compressed} study algorithms in distributed optimization instead of DNES. \\
        \item[$\star$] The convergence accuracy is stated in a mean square sense. The $\omega$-NE represents the asymptotically converged point with a mean square distance of $\omega$ to the NE, where $\omega > 0$. The exact NE corresponds to the $0$-NE.
    \end{tablenotes}
    \end{threeparttable}
    \label{tab: compare}
\end{table}

Our main contributions are as follows:
\begin{itemize}
	\item[1)]
 We propose a novel Compression-based Privacy-preserving DNES (CP-DNES) algorithm (\textbf{Algorithm~\ref{algo: one}}). 
 CP-DNES encodes the messages with fewer bits and masks information by intrinsic random compression errors.
	\item[2)] 
 By developing precise step size conditions, we demonstrate that CP-DNES converges to the accurate NE in the mean square sense, even in the presence of the non-vanishing compression errors (\textbf{Theorem~\ref{thm: convergence}}). 
	\item[3)] 
CP-DNES, when equipped with a specific stochastic compressor, achieves $(0,\delta)$-DP (\textbf{Theorem~\ref{thm: DP}}), surpassing the commonly used $(\epsilon, \delta)$-DP. This result sheds light on simultaneously attaining $(0,\delta)$-DP and convergence accuracy in distributed aggregative games. 
\end{itemize}


\emph{Notations}: 
Let $\mathbb{R}^{p}$ and $\mathbb{R}^{p \times q}$ denote the set of $p$-dimensional vectors and $p \times q$-dimensional matrices, respectively, and $\mathbb{N}_{+}$ represents the set of positive integers.
For $\mathcal{N} \triangleq \{1, 2, \dots, N \}$, $\text{col}(x_{i})_{i \in \mathcal{N}}$ refers to the stacked vector $[x_{1}^{\top}, \dots, x_{N}^{\top}]^{\top}$.
The notations $\mathbf{1}_{p} \in \mathbb{R}^{p}$ denotes a vector with all elements equal to one, and $I_{p} \in \mathbb{R}^{p \times p}$ represents a $p \times p$-dimensional identity matrix.
Denote by $\prod_{i=1}^{N}\mathcal{X}_{i}$ the Cartesian product of the set $\{ \mathcal{X}_{i} \}_{i=1, \dots, N}$.
\
For a closed and convex set $\mathcal{X} \subseteq \mathbb{R}^{n}$, the projection operator $\mathbf{P}_{\mathcal{X}}(\cdot): \mathbb{R}^{n} \to \mathcal{X}$ is defined as $\mathbf{P}_{\mathcal{X}}(v) = \arg\min_{z \in \mathcal{X}} \|v-z \|$.
The operator $\| \cdot \|$ is the induced-$2$ norm for matrices and the Euclidean norm for vectors. 
We use $\| \cdot \|_{1}$ to denote the $\ell_{1}$-norm of a vector $x = \text{col}(x_{i})_{i=1}^{p} \in \mathbb{R}^{p}$, and $\| x \|_{1} = \sum_{i=1}^{p} |x_{i}| $.
We use $\mathbb{P}(\mathcal{A})$ to represent the probability of an event $\mathcal{A}$, and $\mathbb{E}[x]$ to be the expected value of a random variable $x$.
For any two matrices, $A \in \mathbb{R}^{n \times m}$, $B \in \mathbb{R}^{p \times q}$, $A \otimes B \in \mathbb{R}^{np \times mq}$ is the Kronecker product of $A$ and $B$.

Let an undirected graph $\mathcal{G} = (\mathcal{N}, \mathcal{E})$ describe the information exchange among a set of $N$ players, denoted by $\mathcal{N} = \{1, 2, \dots, N \}$. The edge set $\mathcal{E} \subset \mathcal{N} \times \mathcal{N}$ denotes the communication links. The weight matrix $W = [w_{ij}]$ represents the structure of interactions in $\mathcal{G}$. If player $i$ can receive messages from player $j$, then $w_{ij}>0$. Otherwise, $w_{ij} = 0$.
We define the neighbor set of player $i$ as $\mathcal{N}_{i} = \{j| w_{ij}>0\}$, and the degree matrix $D$ is a diagonal matrix with the $i$-th element $d_{ii} = \sum_{j \in \mathcal{N}_{i}}w_{ij}$. 
Therefore, the Laplacian matrix is $L = D-W$.

\section{Preliminaries and Problem Statement} \label{sec: formulation}

\subsection{Attacker and DP}
We consider a prevalent eavesdropping attack model in privacy~\cite{ye2021differentially, wang2022differentially}, where attackers monitor all communication channels to intercept transmitted messages and learn sensitive information from the sending players.

DP quantifies the privacy level of involved individuals in a statistical database. 
We provide the following definitions for DP in distributed aggregative games.
\begin{definition} \label{defn: adjacency} 
	(Adjacency~\cite{ye2021differentially}) Two objective function sets $\mathcal{F} = \{f_{i} \}_{i \in \mathcal{N}}$ and $\mathcal{F}^{\prime} = \{f_{i}^{\prime} \}_{i \in \mathcal{N}}$ are adjacent if there exists some $i_{0} \in \mathcal{N}$ such that $f_{i} = f_{i}^{\prime}, \forall i \neq i_{0}$ and $f_{i_{0}} \neq f_{i_{0}}^{\prime}$.
\end{definition}

\begin{definition} \label{defn: dp} (Differential Privacy~\cite{dwork2006calibrating}) 
For a randomized mechanism $\mathcal{M}$, it preserves $(\epsilon, \delta)$-DP of objective functions if, for any pair of adjacent function sets $\mathcal{F}$ and $\mathcal{F}^{\prime}$ and all subsets $\mathcal{O}$ of the image set of the mechanism $\mathcal{M}$, the following condition holds:
\begin{equation} \label{eq: dp_def}
	\mathbb{P}[\mathcal{M}(\mathcal{F}) \in \mathcal{O}] \leq e^{\epsilon} \mathbb{P}[\mathcal{M}(\mathcal{F}^{\prime}) \in \mathcal{O}] + \delta,
\end{equation}
where $\epsilon \geq 0$ and $\delta \geq 0$.
\end{definition}

Definition~\ref{defn: dp} states that changing the cost function of a single player leads to a small change in the distribution of outputs.
The factor $\epsilon$ in~\eqref{eq: dp_def} denotes the privacy upper bound to measure an algorithm $\mathcal{M}$, and $\delta$ denotes the probability of breaking this bound. Therefore, a smaller $\epsilon$ and $\delta$ indicate a higher level of privacy.

\subsection{Problem Formulation}
Consider a network of $N$ players solving a distributed aggregative game. 
Let $x_{i} \in \mathcal{X}_{i}$ denote the decision of player $i$ and $\mathcal{X}_{i} \subset \mathbb{R}^{n}$ is a bounded closed convex set.
The player $i$'s cost function $f_{i}(x_{i}, h(\mathbf{x})): \mathcal{X}  \to \mathbb{R}$ depends on its decision $x_{i}$ and an aggregate function of all players' decisions $h(\textbf{x}) = \frac{1}{N} \sum_{i=1}^{N}x_{i}$, where $\mathcal{X} = \prod_{i \in \mathcal{N}}\mathcal{X}_{i}$.
The following optimization problem represents the objective of each player in the game:
\begin{equation} \label{eq: problem}
\begin{aligned}
	\min_{x_{i} \in \mathcal{X}_{i}} 
	&\ f_{i}(x_{i}, h(\mathbf{x})) = f_{i}\left(x_{i}, \frac{1}{N}x_{i} + h(\mathbf{x}_{-i})\right), \ \forall i \in \mathcal{N},
\end{aligned}
\end{equation}
where $h(\mathbf{x}_{-i}) = \frac{1}{N} \sum_{j \in \mathcal{N}, j \neq i} x_{j}$ is the aggregative value of player $i$'s opponents.
The optimal solution to~\eqref{eq: problem} is the NE.
\begin{definition}
	A Nash equilibrium (NE) is a decision profile $\mathbf{x}^{*} = \text{col}(x_{i}^{*})_{i \in \mathcal{N}} \in \mathcal{X} $ satisfying the following condition for each player $i \in \mathcal{N}$: $f_{i}(x_{i}^{*}, \frac{1}{N} x_{i}^{*} + h(\mathbf{x}_{-i}^{*})) \leq f_{i}(x_{i}, \frac{1}{N}x_{i} + h(\mathbf{x}_{-i}^{*}))$.
\end{definition}

To ensure that players can find the NE through interactions, we adopt the following assumption: 
\begin{assum} \label{assum: graph}
	The communication network is an undirected and connected graph.
\end{assum}

The following assumptions are posed on each objective function $f_{i}, \forall i \in \mathcal{N}$.
\begin{assum} \label{assum: existence}
The function $f_{i}(x_{i}, h(\mathbf{x}))$ is convex in $x_{i}$ for every fixed $x_{-i} \in \prod_{j\in \mathcal{N}, j\neq i} \mathcal{X}_{j}$ and $f_{i}(x_{i},v)$ is continuously differentiable in $(x_{i}, v) \in \mathcal{X}_{i} \times \mathbb{R}^{n}$.
\end{assum}
Under Assumption~\ref{assum: existence}, we define the following for player $i$:
\begin{subequations}
	\begin{align}
		\label{eq: g} g_{i}(x_{i},z_{i}) \triangleq & \left( \nabla_{x_{i}}  f(x_{i}, v) + \frac{1}{N} \nabla_{v} \partial f(x_{i}, v) \right) \bigg|_{v = z_{i}}, \\
		\phi_{i}(\mathbf{x}) \triangleq & \nabla_{x_{i}} f_{i}(x_{i}, h(\mathbf{x})). 
	\end{align}
\end{subequations}
Denote $\mathbf{z} = \text{col}(z_{i})_{i \in \mathcal{N}}$, $G(\mathbf{x}, \mathbf{z}) = \text{col}(g_{i}(x_{i}, z_{i}))_{i \in \mathcal{N}}$ and $\Phi(\mathbf{x}) = \text{col}( \phi_{i}(\mathbf{x}) )_{i \in \mathcal{N}}$. It can be easily inferred that $\phi_{i}(\mathbf{x}) = g_{i}(x_{i}, h(\mathbf{x}))$ and $\Phi(\mathbf{x}) = G(\mathbf{x}, \mathbf{1}_{N}\otimes I_{n} h(\mathbf{x}) )$.

\begin{assum} \label{assum: unique} 
The gradient $\phi_{i}(\mathbf{x})$ is $L_{\phi}$-Lipschitz continuous, i.e., $\| \phi_{i}(\mathbf{x}) - \phi_{i}(\mathbf{y})  \| \leq L_{\phi} \left\| \mathbf{x} - \mathbf{y} \right\|$ for $\mathbf{x}, \mathbf{y} \in \mathbb{R}^{n}$. Moreover, there exists $m>0$ such that for $\mathbf{x} \in \mathbb{R}^{Nn}$, $(\mathbf{x} - \mathbf{x}^{*})^{\top}(\Phi(\mathbf{x}) - \Phi(\mathbf{x}^{*})) \geq m \left\| \mathbf{x} - \mathbf{x}^{*} \right\|^{2}$ holds.
\end{assum}

Assumptions~\ref{assum: existence} and~\ref{assum: unique} are standard for guaranteeing the existence and uniqueness of the NE in problem~\eqref{eq: problem}~\cite{koshal2016distributed, huang2022linearly, zhu2021networked}, which is crucial for designing DNES. If either of these assumptions is not satisfied, the existence and uniqueness of the NE need to be rigorously analyzed~\cite{wang2022non}. However, this is outside the scope of this paper since we mainly focus on the algorithm design here.


\begin{assum} \label{assum: Lip}
	For any $x_{i} \in \mathcal{X}_{i}$, $g_{i}(x_{i}, z)$ is $L_{g}$-Lipschitz in $z$, i.e., $\| g_{i}(x_{i}, z_{1}) - g_{i}(x_{i}, z_{2})  \| \leq L_{g} \| z_{1} - z_{2} \|, \ \forall z_{1}, z_{2} \in \mathbb{R}^{n}$.
\end{assum}

Assumption~\ref{assum: Lip} is commonly used in literature on aggregative games~\cite{koshal2016distributed, huang2022linearly, zhu2021networked}, 
crucial for convergence analysis and encompassing numerous prevalent game models, such as auction-based games~\cite{salehisadaghiani2016distributed} and Cournot oligopoly games~\cite{wang2021distributed}. Moreover, it should be noted that $L_{\phi}$, $m$, and $L_{g}$ in Assumptions~\ref{assum: existence}--\ref{assum: Lip} are solely employed for theoretical analysis. In the algorithm implementation, players do not need precise value of these parameters.

\begin{assum} \label{assum: bounded_gradient}
 There exists a positive constant $C$ such that there is $\| g_{i} \| \leq C$ and $\| \phi_{i} \| \leq C$, $\forall i \in \mathcal{N}$. 
\end{assum}

Assumption~\ref{assum: bounded_gradient} holds for various game models~\cite{salehisadaghiani2016distributed, wang2021distributed} since we consider a bounded constraint set $\mathcal{X}$. In practice, gradient clipping is commonly used to ensure a bounded gradient.


We aim to design a DNES algorithm that preserves the $(\epsilon, \delta)$-DP of objective functions and converges to the exact NE of~\eqref{eq: problem} in a mean square sense as shown in Definition~\ref{defn: mean}.
\begin{definition} \label{defn: mean}
    With a DNES algorithm, denote the players' action profile at iteration $k$ as $\mathbf{x}_{k}$. The DNES algorithm converges to the exact NE, $\mathbf{x}^{*}$, in a mean square if
    \begin{equation}
        \lim_{k \to \infty} \mathbb{E} \left[\| \mathbf{x}_{k} - \mathbf{x}^{*} \|^{2} \right] =0.
    \end{equation}
\end{definition}

\section{Algorithm Development} \label{sec: algorithm}
Conventional DNES for aggregative games typically have the following form~\cite{ye2017distributed}:
\begin{subequations} \label{eq: conventional}
\begin{align}
	x_{i,k+1} &= \mathbf{P}_{\mathcal{X}_{i}}[ x_{i,k} - \eta g_{i}(x_{i,k},y_{i,k})], \\
 \label{eq: dynamic_consensus}	y_{i,k+1} &= y_{i,k} + \sum_{j \in \mathcal{N}_{i}} w_{i,j}(y_{j,k} - y_{i,k}) + x_{i,k+1} - x_{i,k},
\end{align}
\end{subequations}
where $y_{i,k}$ is the estimate of the current aggregative variable $h(\mathbf{x}_{k}) = \frac{1}{N}\sum_{i=1}^{N}x_{i,k}$
by player $i$, and $\eta > 0$ is the step size to optimize the strategy.
To accurately estimate the unknown 
$h(\mathbf{x}_{k})$, player $i$ should broadcast its estimate $y_{i,k}$ to its neighbors and employ dynamic average consensus as shown in~\eqref{eq: dynamic_consensus} to track the aggregative term.
With the public knowledge of $W$ and the step size $\eta$, an attacker can calculate $x_{i,k+1} - x_{i,k}$ and consequently obtain $g_{i}(x_{i,k}, y_{i,k})$ if $x_{i,k} - \eta g_{i}(x_{i,k}, y_{i,k}) \in \mathcal{X}_{i}$.

Based on this observation, we leverage stochastic compression to preserve privacy:
\begin{subequations} \label{eq: algorithm}
	\begin{align}
	\label{eq: x_update} x_{i,k+1} =& \mathbf{P}_{\mathcal{X}_{i}}[ x_{i,k} - \alpha_{k} \beta_{k} g_{i}(x_{i,k},y_{i,k})], \\
	\label{eq: y_update} y_{i,k+1} =& y_{i,k} + \beta_{k} \sum_{j \in \mathcal{N}_{i}} w_{i,j}(\mathcal{C}(y_{j,k}) - \mathcal{C}(y_{i,k})) \nonumber \\
	& \quad \quad + x_{i,k+1} - x_{i,k},
\end{align}
\end{subequations}
where $y_{i,0} = x_{i,0}$, $\mathcal{C}(\cdot)$ is the compression operator, $\alpha_{k}, \beta_{k} >0$ are designed step sizes. We consider stochastic compression schemes that satisfy the following assumption. An example of the compression scheme will be introduced in Section~\ref{sec: DP}.	
\begin{assum} \label{assum: compressor}
For some constant $\sigma$ and any $x \in \mathbb{R}^{n}$, 
  $\mathbb{E}[\mathcal{C}(x)|x ] = x$ and $\mathbb{E}[ \| \mathcal{C}(x) - x \|^{2} ] \leq \sigma^{2}$. The randomized mechanism in each player's compression is statistically independent.
\end{assum}

In Algorithm~\ref{algo: one}, each player shares their compressed estimate $\mathcal{C}(y_{i,k})$ to its neighbors. 
Owing to the independent random compression error, an attacker cannot deduce the exact gradient of player $i$, even with the knowledge of the compression scheme. Thus privacy is preserved.
\setlength{\lineskip}{0pt}
\begin{figure}[t]
  \begin{algorithm}[H]
	\caption{
	CP-DNES}
	\begin{algorithmic}[1] \label{algo: one}
	\renewcommand{\algorithmicrequire}{\textbf{Input:}}
    \renewcommand{\algorithmicensure}{\textbf{Initialize:}}
    \REQUIRE Public information $W$, $\alpha_{k}$, $\beta_{k}$, the total number of iterations $T$
    \ENSURE  $x_{i,0} \in \mathbb{R}^{n} $, $y_{i,0} = x_{i,0}$.
	\FOR {$k = 0, 1, 2,  \dots$}
	\STATE for each $i \in \mathcal{N}$,
	\STATE \quad Compute local gradient $g_{i}(x_{i,k}, y_{i,k})$ using~\eqref{eq: g}. \\
	\STATE \quad Determine compressed estimate $\mathcal{C}(y_{i,k})$ and send it to its neighbors.\\
	\STATE \quad Receive $\mathcal{C}(y_{j,k})$ from $j \in \mathcal{N}$ and update its decision and local estimate on aggregative value using~\eqref{eq: x_update} and~\eqref{eq: y_update}, respectively.
	\ENDFOR
	\end{algorithmic}
  \end{algorithm}
\end{figure}

Let $\mathbf{x}_{k} = \text{col}(x_{i,k})_{i \in \mathcal{N}}$, $\mathbf{y}_{k} = \text{col}(y_{i,k})_{i \in \mathcal{N}}$. 
Then, we can express~\eqref{eq: algorithm} in a compact form:
\begin{subequations}
\begin{align}
	\label{eq: x_compact} \mathbf{x}_{k+1} =& \mathbf{P}_{\mathcal{X}}[ \mathbf{x}_{k} - \alpha_{k}\beta_{k}G(\mathbf{x}_{k}, \mathbf{y}_{k})], \\
	\label{eq: y_compact} \mathbf{y}_{k+1} =& (A_{k} \otimes I_{n} ) \mathbf{y}_{k} + \mathbf{x}_{k+1} - \mathbf{x}_{k} - \beta_{k}(L \otimes I_{n} )\mathbf{e}_{k},
\end{align}
\end{subequations}
where $A_{k} = I - \beta_{k}L$, $\mathbf{e}_{k} = \text{col}(e_{i,k})_{i\in \mathcal{N}}$, and $e_{i,k} = \mathcal{C}(y_{i,k}) - y_{i,k}$ is the compression error. 
It can be inferred that $A_{k}$ is doubly stochastic.

The dynamics of $\mathbf{y}_{k}$ in~\eqref{eq: y_compact} imply that the average estimate by players, $\bar{\mathbf{y}}_{k} = \frac{1}{N}(\mathbf{1}_{N}^{\top} \otimes I_{n} )\mathbf{y}_{k}$, 
evolves according to the dynamics $\bar{\mathbf{y}}_{k+1} = \bar{\mathbf{y}}_{k} + \bar{\mathbf{x}}_{k+1} - \bar{\mathbf{x}}_{k}$,
where $\bar{\mathbf{x}}_{k} = \frac{1}{N}(\mathbf{1}_{N}^{\top} \otimes I_{n}) \mathbf{x}_{k} = h(\mathbf{x}_{k})$. 
Therefore, we have $\mathbf{1}\bar{\mathbf{y}}_{k+1} - \mathbf{1}\bar{\mathbf{x}}_{k+1} = \mathbf{1}\bar{\mathbf{y}}_{k} - \mathbf{1}\bar{\mathbf{x}}_{k} = \cdots =  \mathbf{1}\bar{\mathbf{y}}_{0} - \mathbf{1}\bar{\mathbf{x}}_{0}=0$, i.e., 
\begin{equation} \label{eq: y_bar}
	\bar{\mathbf{y}}_{k} = \bar{\mathbf{x}}_{k} = h(\mathbf{x}_{k}).
\end{equation}
 This shows that the evolution of $\bar{\mathbf{y}}_{k}$ is immune to the quantization error $\mathbf{e}_{k}$, and thus $\mathbf{y}_{k}$ tracks the aggregative value $h(\mathbf{x}_{k})$, which helps to ensure convergence accuracy.

\begin{rem} 
To preserve privacy in conventional DNES~\eqref{eq: conventional}, we introduce random perturbation via stochastic compression schemes. However, quantifying the privacy level arising from compression and ensuring convergence without sacrificing accuracy pose new challenges unaddressed by existing DNES algorithms~\cite{ye2021differentially, lin2023statistical, chen2023differentially, chen2022distributed, pei2023distributed, chen2022linear}.
\end{rem}
\begin{rem}
Under Assumption~\ref{assum: compressor}, the compression error of each player is independently and identically distributed (i.i.d), which avoids leakage of the exact original value, like the i.i.d noise perturbation approach, thereby preserving message privacy.
Some works adopt Laplacian or Gaussian noise with decreasing variance to protect the privacy~\cite{ye2021differentially, lin2023statistical}. 
	However, this approach poses a high risk of privacy leakage when approaching convergence. 
	In contrast, the variance of compression error in Algorithm~\ref{algo: one} does not decrease over time. 
	However, the non-decreasing compression error introduces convergence challenges. To tackle these, we meticulously design the step sizes to guarantee convergence.
\end{rem}



\section{Convergence Analysis} \label{sec: convergence}
In this section, we demonstrate that CP-DNES ensures convergence of all players' decisions to $\mathbf{x}^{*}$ under certain conditions. 
For simplicity, we assume $n=1$ to simplify the convergence analysis, and it can be easily extended to the case when $n>1$.
 

\begin{lem} \label{lem: y_bounded}
Suppose Assumptions~\ref{assum: graph}--\ref{assum: compressor} hold, then 
	$ \sum_{k=0}^{\infty} \beta_{k} \mathbb{E} [\| \mathbf{y}_{k} - \mathbf{1} \bar{\mathbf{x}}_{k}\|^{2}] < \infty $ 
	if $\sum_{k=0}^{\infty} \alpha_{k}^{2} \beta_{k} < \infty$, and $\sum_{k=0}^{\infty} \beta_{k}^{2} < \infty$.
\end{lem}
\begin{proof}
	The proof is provided in Appendix~\ref{app: lem1}.
\end{proof}

Lemma~\ref{lem: y_bounded} indicates that if the step sizes satisfy certain conditions, the cumulative disagreement of the estimates on the actual estimate, $\bar{\mathbf{x}}_{k}$, will not diverge over time.
Based on this preliminary result, we can prove the convergence of CP-DNES. 
\begin{thm} \label{thm: convergence}
Under Assumptions~\ref{assum: graph}--\ref{assum: compressor}, if the step size sequence satisfies
\begin{equation} \label{eq: step}
	\sum_{k=0}^{\infty} \alpha_{k} \beta_{k} = \infty, \ \sum_{k=0}^{\infty} \alpha_{k}^{2} \beta_{k} < \infty, \ \sum_{k=0}^{\infty} \beta_{k}^{2} < \infty,
\end{equation}
then CP-DNES guarantees that the sequence $\{\mathbf{x}_{k}\}$ will converge to $\mathbf{x}^{*}$ in a mean square sense, i.e., 
	\begin{equation}
		 \lim_{k \to \infty} \mathbb{E} \left[\| \mathbf{x}_{k} - \mathbf{x}^{*} \|^{2} \right] =0.
	\end{equation}
\end{thm}
\begin{proof}
We have $\sum_{k=0}^{T} m \alpha_{k} \beta_{k} \mathbb{E} [ \| \mathbf{x}_{k} - \mathbf{x}^{*} \|^{2} ] \leq \| \mathbf{x}_{0} - \mathbf{x}^{*} \|^{2} - \mathbb{E} [ \| \mathbf{x}_{T+1} - \mathbf{x}^{*} \|^{2} ] + 4C^{2}N \sum_{k=0}^{T} \alpha_{k}^{2} \beta_{k}^{2}  + \frac{{L}_{g}^{2}}{m} \sum_{k=0}^{T} \alpha_{k}\beta_{k} \mathbb{E} [ \| \mathbf{y}_{k} - \mathbf{1} \bar{\mathbf{x}}_{k} \|^{2}]$. Since the right-hand side of this relationship is always bounded when $T \to \infty$, we can conclude that $ \mathbb{E} [ \|\mathbf{x}_{k} -  \mathbf{x}^{*}\|^{2} ]$ converges to zero. The detailed proof is provided in Appendix~\ref{app: thm1}.
\end{proof}

\begin{cor} \label{cor: rate}
	If $\alpha_{k} = \frac{c_{1}}{(c_{2}k+1)^{\omega_{1}}}$ and $\beta_{k} = \frac{c_{3}}{(c_{2}k+1)^{\omega_{2}}}$, where $c_{1}, c_{2}, c_{3}, \omega_{1}, \omega_{2} >0$, and $\omega_{1}$ and $\omega_{2}$ satisfy $\omega_{1} + \omega_{2}\leq 1$, $\omega_{2}>0.5$ and $2\omega_{1} + \omega_{2} >1$, then the convergence rate of $ \mathbb{E} [ \|\mathbf{x}_{k} -  \mathbf{x}^{*}\|^{2} ]$ is given by
	\begin{equation}
		\frac{\sum_{k=0}^{T} m \alpha_{k} \beta_{k} \mathbb{E} \left[ \| \mathbf{x}_{k} - \mathbf{x}^{*} \|^{2} \right]}{\sum_{k=0}^{T} \alpha_{k} \beta_{k}} = O\left(\frac{1}{(1+T)^{\omega}} \right),
	\end{equation}
	where $\omega = \min\{ 2 \omega_{1}, \omega_{2} \}$.
\end{cor}
\begin{proof}
	See Appendix~\ref{app: cor1}.
\end{proof}

\begin{rem}
The conventional DNES algorithm~\eqref{eq: conventional} achieves linear convergence to $\mathbf{x}^{*}$ with a constant step size. 
However, in CP-DNES, we leverage random compression errors in Assumption~\ref{assum: compressor} for privacy, potentially causing significant error accumulation with constant step sizes. 
To ensure both privacy and accurate convergence concurrently, we introduce diminishing step sizes $\alpha_{k}$ and $\beta_{k}$. 
Specifically, the diminishing $\beta_{k}$ in~\eqref{eq: x_update} and~\eqref{eq: y_update} mitigates the accumulation of compression errors affecting both consensus and gradient descent during convergence.
\end{rem}



\begin{rem} 
Assumption~\ref{assum: compressor} focuses on the \underline{u}nbiased \underline{c}ompressor with \underline{b}ounded \underline{a}bsolute compression error (UC-BA). 
	Some studies explore \underline{u}nbiased \underline{c}ompressors with \underline{b}ounded \underline{r}elative compression error (UC-BR)~\cite{chen2022linear}, characterized by $\mathbb{E}[\mathcal{C}(x)|x ] = x$ and $\mathbb{E}[ \| \mathcal{C}(x) - x \|^{2} ] \leq \varphi \|x \|^{2}$, with $\varphi > 0$. 
	It is noteworthy that UC-BA and UC-RB differ, and neither is more general than the other one~\cite{yi2022communication}. 
	Despite this, only Lemma 1 employs Assumption 6 in the convergence analysis, and we can easily prove the same results for UC-BR. Hence, Algorithm 1 maintains convergence to the exact NE in a mean square sense even under UC-BR.
\end{rem}
\begin{rem}
    In addition to UC-BA and UC-BR, there exists another type of compressor, namely \underline{b}iased \underline{c}ompressors (BC). However, BC introduces convergence errors, hindering accurate convergence. 
	Moreover, quantifying the DP level necessitates a specific error distribution.
The error distribution of UC-BR depends on the input distribution, and the error feedback commonly used for BC further complicates this distribution. Thus, investigating privacy-preserving DNES algorithms with UC-BR and BC remains a challenge. A rough idea to solve this issue is to adapt the compression scheme or design consensus approaches to control the compression error distribution~\cite{hegazy2023compression}. We leave the details as future work.
\end{rem}

\section{Differential Privacy Analysis} \label{sec: DP}
In this section, we will prove that CP-DNES ensures rigorous DP under a specific stochastic compression scheme~\cite{aysal2008distributed}.

\begin{definition} \label{defn: random_quantization}
	The element-wise stochastic compressor quantizes a vector $x = [x_{(1)}, x_{(2)}, \dots, x_{(n)}]^{\top} \in \mathbb{R}^{n}$ to the range representable by a scale factor $\theta \in \mathbb{N}_{+}$ and $b$ bits, $\{ -2^{b-1}\theta, \dots, -\theta, 0, \theta, \dots, ( 2^{b}-1 )\theta \}$, where $2^{b}\theta > |x_{(i)}|,\ \forall i =1, \dots, n$. In other words, $\mathcal{C}(x) = [\mathcal{C}(x_{(1)}), \mathcal{C}(x_{2}), \dots, \mathcal{C}(x_{(d)})]^{\top}$.
	For any $l\theta \leq x_{(i)} < (l+1)\theta$, the compressor outputs
	\begin{equation*} 
	\mathcal{C}(x_{(i)}) = \left\{
	\begin{array}{cl}
		l\theta, & \text{with probability}\ 1+l-x_{(i)}/\theta, \\
		(l+1)\theta, & \text{with probability}\ x_{(i)}/\theta - l. \\
	\end{array}
	\right.
\end{equation*}
\end{definition}

Definition~\ref{defn: random_quantization} is equivalent to dithered quantization in signal processing and it is illustrated in Fig.~\ref{fig: random_quantization}.
If the observed value by an attacker is $\theta$, the original data could take any value within the range $[0, 2 \theta)$. When increasing the scale factor $\theta$, we need fewer bits to ensure $2^{b}\theta > |x_{(i)}|$.
\begin{figure}[t]  
	\centering
	\includegraphics[width=0.45\linewidth]{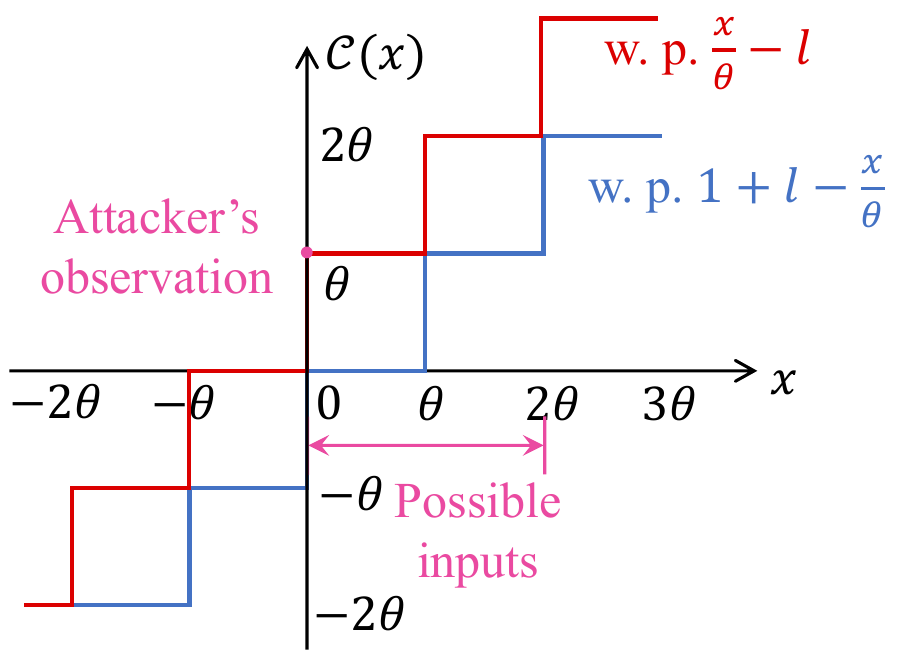}
	\caption{Illustration of the stochastic compressor in Definition~\ref{defn: random_quantization}.} 
	\label{fig: random_quantization}
\end{figure}

\begin{thm} \label{thm: DP}
	Under Assumptions~\ref{assum: graph}--\ref{assum: compressor}, if $\alpha_{k}$ and $\beta_{k}$ satisfy $\alpha_{k}\beta_{k} = \frac{c_{4}}{c_{5}k+1}$, where $c_{4}, c_{5}>0$, and the conditions in~\eqref{eq: step}, the stochastic compressor defined in Definition~\ref{defn: random_quantization} ensures the preservation of $(0, \delta_{k})$-DP for the objective function of each player at the $k$-th iteration while guaranteeing convergence, where 
	\begin{equation} \label{eq: privacy_level}
		\delta_{k} = \min \left\{ 1, \frac{2Cc_{4}\sqrt{n}}{c_{5}\theta} \ln(c_{5}k+1) \right \}.
	\end{equation}
\end{thm}
\begin{proof}
For any pair of adjacent function sets $\{ \mathcal{F}\}$ and $\{ \mathcal{F}^{\prime}\}$, the eavesdropper has identical observations. Thus, we obtain that $\delta_{k} = \mathbb{P}[\mathcal{C}(y_{i(j)}) = l \theta | y_{i}] - \mathbb{P}[\mathcal{C}(y_{i(j)}^{\prime}) = l \theta | y_{i}^{\prime}] \leq  \frac{\| \Delta y_{i_{0},k} \|_{1}}{\theta}$ and  $\| \Delta y_{i_{0},k} \|_{1} \leq \sqrt{n} \| \Delta y_{i_{0},k} \| \leq \frac{2Cc_{4}\sqrt{n}}{c_{5}} \ln(c_{5}k+1)$.
Additionally, it is important to note that in DP, $\delta_{k}$ should be a small parameter within the range of $(0,1)$.
Hence, we derive the expression of $\delta_{k}$ shown in~\eqref{eq: privacy_level}.
The detailed proof is provided in Appendix~\ref{app: thm2}.
\end{proof}

\begin{rem}
	The ternary quantization scheme in~\cite{wang2022quantization} is a special case of Definition~\ref{defn: random_quantization}, which compresses each elements to $\{ -\theta, 0, \theta \}$ with $1$ bit, i.e., $b=1$, and $\theta$ should be larger than any possible values of $y_{i(j)}$. In this case, the privacy budget $\delta_{k} = \frac{2Cc_{4}\sqrt{n}}{c_{5}\theta} \ln(c_{5}k+1)$, which is a very small value.
\end{rem}             


\begin{rem}
	For the compressors in Assumption~\ref{assum: compressor}, Yi et al.~\cite{yi2022communication} and Liao et al.~\cite{liao2023linearly} transmit $r_{k}\mathcal{C}(y_{i,k}/ r_{k})$, where $r_{k}$ exponentially decays to decrease the compression error variance over time and achieve convergence. 
	However, this scaling method leads to an exponential growth of the privacy budget, which is $\frac{\| \Delta y_{i_{0},k} \|_{1}}{r_{k}\theta}$ per iteration. 
	Consequently, it fails to provide adequate privacy protection. Instead of using the scaling method to enable convergence, our algorithm adopts a carefully designed step size to ensure both convergence and DP simultaneously, considering the compressors specified in Definition~\ref{defn: random_quantization}.
\end{rem}
\begin{rem} 
    We utilize Assumption~\ref{assum: bounded_gradient} in convergence analysis and rely on the value of the gradient bound, $C$, to quantify the privacy level. Many studies employ bounded gradient for both convergence analysis and privacy quantification~\cite{ye2021differentially, wang2022quantization, lin2023statistical}.
    In future, it would be intriguing to employ induction techniques and consider specific adjacent objective function sets to eliminate the requirement of bounded gradients~\cite{huang2024differential}.
\end{rem}


\section{Numerical Simulations} \label{sec: sim}
We consider energy consumption games in a network of heating, ventilation, and air conditioning systems~\cite{ye2017distributed}, where five end-users communicate using a ring topology.
The objective function of user $i$ is $f_{i}(\mathbf{x}) = (x_{i} - s_{i})^{2} + (p_{0} \sum_{j=1}^{N} x_{j} + h )x_{i}$, $\mathcal{X}_{i} \in [30,50],\ \forall i \in \mathcal{N}$, where $s_{i}$ denotes the energy required to regulate indoor temperature, and $p_{0} \sum_{i=j}^{N} x_{j} + h$ represents the price. 
We set $s_{1} = 56$, $s_{2} = 40$, $s_{3} = 43$, $s_{4} = 60$, $s_{5} = 50$, $p_{0} = 0.05$ and $h=8$. 
Through centralized calculation, we determine that the aggregative game has a unique NE $\mathbf{x}^{*} = [45.8749, 30.2651, 33.1919, 49.7773, 40.0212]^{\top}$.

To evaluate the performance of CP-DNES, we set $\alpha_{k} = \frac{0.4}{(k+1)^{0.3}}$ and $\beta_{k} = \frac{0.4}{(k+1)^{0.6}}$. The estimate satisfies $|y_{i,k}|<90, \ \forall k \geq 0$, and $C=15$.
We employ the compressor given in Definition~\ref{defn: random_quantization}, with $\sigma = \frac{\theta}{2}$ in Assumption~\ref{assum: compressor}.
Therefore, the number of bits per agent per iteration is $b = \lceil \log_{2}(90/\theta) \rceil$ and CP-DNES preserves $(0, \delta_{k})$-DP with $\delta_{k} = \min \{ 1, \frac{4.8}{\theta} \ln (k+1)  \}$.
Due to the randomness of the compressor, we conducted the simulation $100$ times to obtain the empirical mean.

\subsection{Convergence-Communication-Privacy Trade-off}
We simulate CP-DNES using various compression parameters, $\mathbb{C}_{1}$, $\mathbb{C}_{2}$, and $\mathbb{C}_{3}$, as detailed in Table~\ref{tab: params}, and compare it with conventional DNES~\cite{ye2017distributed} under the same step sizes.
The comparison results are depicted in Fig.~\ref{fig: com}.
\begin{table}[t] 
    \centering
\caption{\small{Different compression parameters with their compression errors and privacy levels}}
    \begin{tabular}{|c|c|c|c|c|}
        \hline
         & $\theta$ & $b$ & $\sigma$ & $\delta_{k}$ \\
        \hline
        $\mathbb{C}_{1}$ & $10$ & $4$ & $5$ & $\delta_{k} = \min\{1, 0.48 \ln(k+1)\}$ \\
        \hline
        $\mathbb{C}_{2}$ & $40$ & $2$ & $20$  & $\delta_{k} = \min\{1, 0.12 \ln(k+1)\}$ \\
        \hline
        $\mathbb{C}_{3}$ & $60$ & $1$ & $30$& $\delta_{k} = \min\{1, 0.08 \ln(k+1)\}$ \\
        \hline
    \end{tabular}
    \label{tab: params}
\end{table}
Fig.~\ref{fig: con_delta} shows that players' decisions converge to the NE asymptomatically under CP-DNES and CP-DNES converges slower than the conventional DNES, with fewer transmitted bits. 
However, conventional DNES lacks any privacy protection, and Fig.~\ref{fig: com_res} illustrates the total transmitted bits required to achieve a certain quality of NE. 
Comparatively, conventional DNES transmitting floats using 32 bits per player per iteration consumes significantly more communication resources than CP-DNES.
Therefore, despite slower convergence, our algorithm significantly saves communication resources in obtaining the same NE quality.

A trade-off exists between the value of $b$ and the convergence rate, both of which impact the communication costs for attaining a certain quality of NE. 
	The compression parameter $\mathbb{C}_{2}$ uses fewer bits than $\mathbb{C}_{1}$ but converge to the NE with $\mathbb{E}[\|x_{k}-x^{*} \|] \leq 0.08$ more quickly than $\mathbb{C}_{3}$.
	Thus, as illustrated in Fig.~\ref{fig: com_res}, $\mathbb{C}_{2}$ consumes the least communication resources.
\begin{figure}[t]
	\centering
	\begin{subfigure}[t]{0.45\linewidth}
		\includegraphics[width=\linewidth]{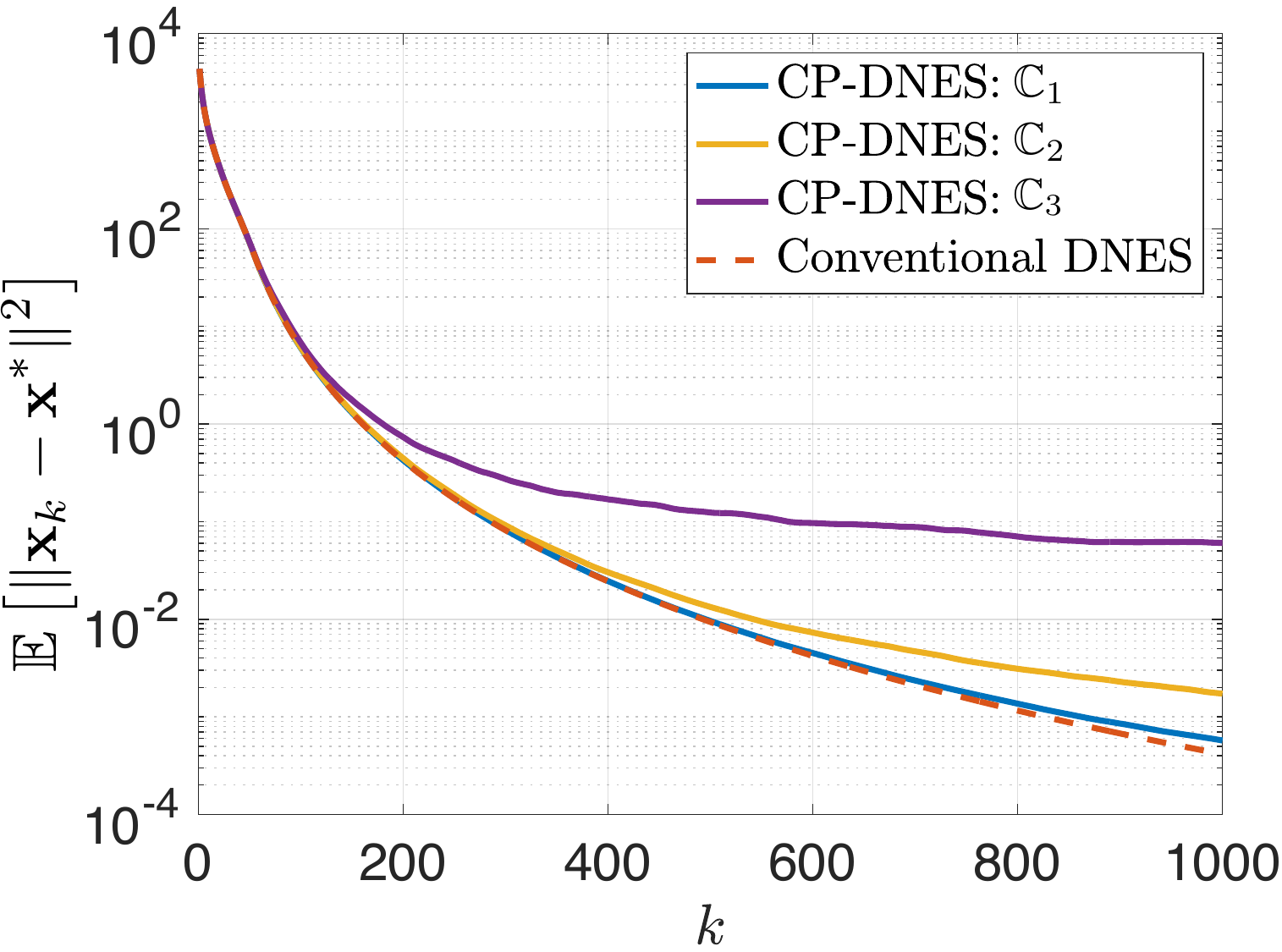}
  \captionsetup{font = small}
		\caption{Convergence comparison.}
		\label{fig: con_delta}
	\end{subfigure}
	\begin{subfigure}[t]{0.45\linewidth}
		\includegraphics[width=\linewidth]{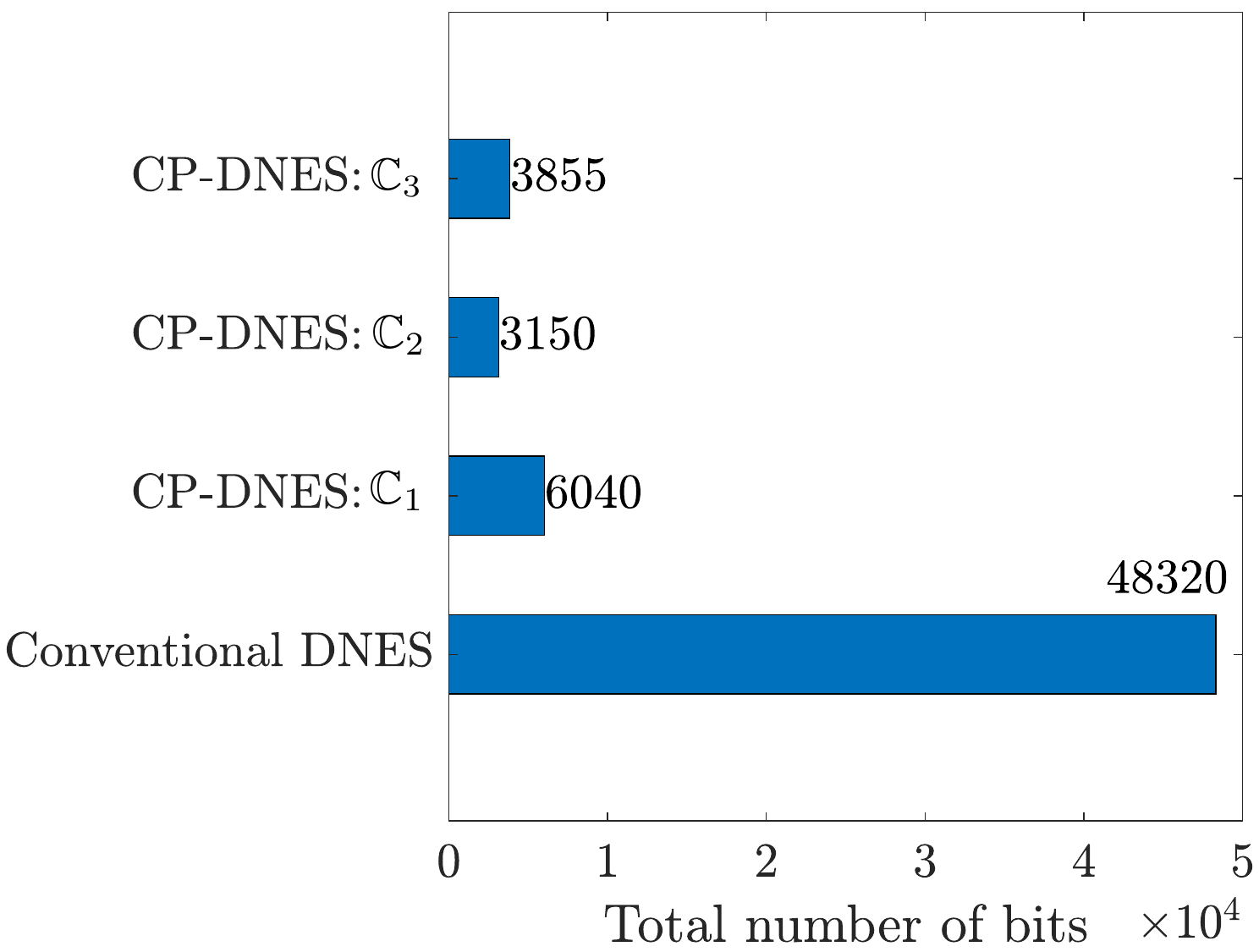}
  \captionsetup{font = small}
		\caption{Total transmitted bits to obtain $\mathbb{E}[\|x_{k}-x^{*} \|^{2}] \leq 0.08$.} 
		\label{fig: com_res} 
	\end{subfigure}
\caption{\small {Comparison between conventional DNES without compression-based privacy preservation~\cite{ye2017distributed} and CP-DNES with different compression parameters shown in Table~\ref{tab: params}. }}
\label{fig: com}
\end{figure}


\subsection{Comparison with State-of-the-Art}
We compare our proposed algorithm with some existing DNES algorithms, including NP-DNES (DNES with noise perturbation)~\cite{ye2021differentially}
and DSC-DNES (DNES with dynamic scaling compression used in~\cite{liao2023linearly, yi2022communication}) in Fig.~\ref{fig: compare}. 
The variance of the noise at iteration $k$ in NP-DNES is $0.91^{k}$. We set the dynamic factor in DSC-DNES as $r_{k}=0.87^{k}$, with each player using 8 bits to transmit the message per iteration.
Fig.~\ref{fig: compare} demonstrates that CP-DNES with $\mathbb{C}_{3}$ has superior convergence performance even if each player transmits only 1 bit per iteration, and requires the fewest bits to achieve a specific quality of NE.
In NP-DNES, noise perturbation affects accuracy. In DSC-DNES, dynamic scaling amplifies compression error under the specific communication constraint. 

\begin{figure}[t]
	\centering
	\begin{subfigure}[t]{0.45\linewidth}
		\includegraphics[width=\linewidth]{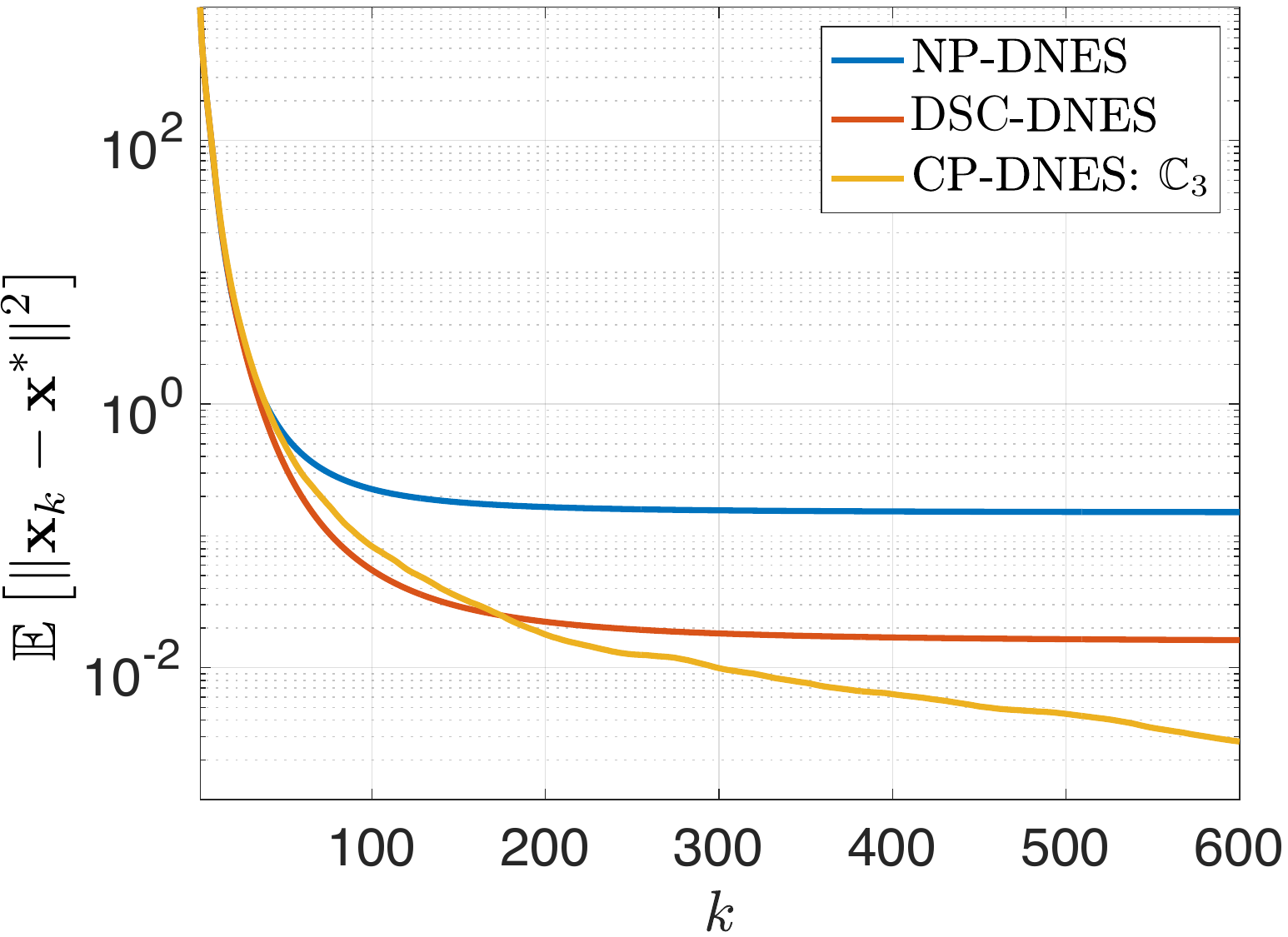}
  \captionsetup{font = small}
		\caption{Convergence comparison.}
		\label{fig: compare_con}
	\end{subfigure}
	\begin{subfigure}[t]{0.45\linewidth}
		\includegraphics[width=\linewidth]{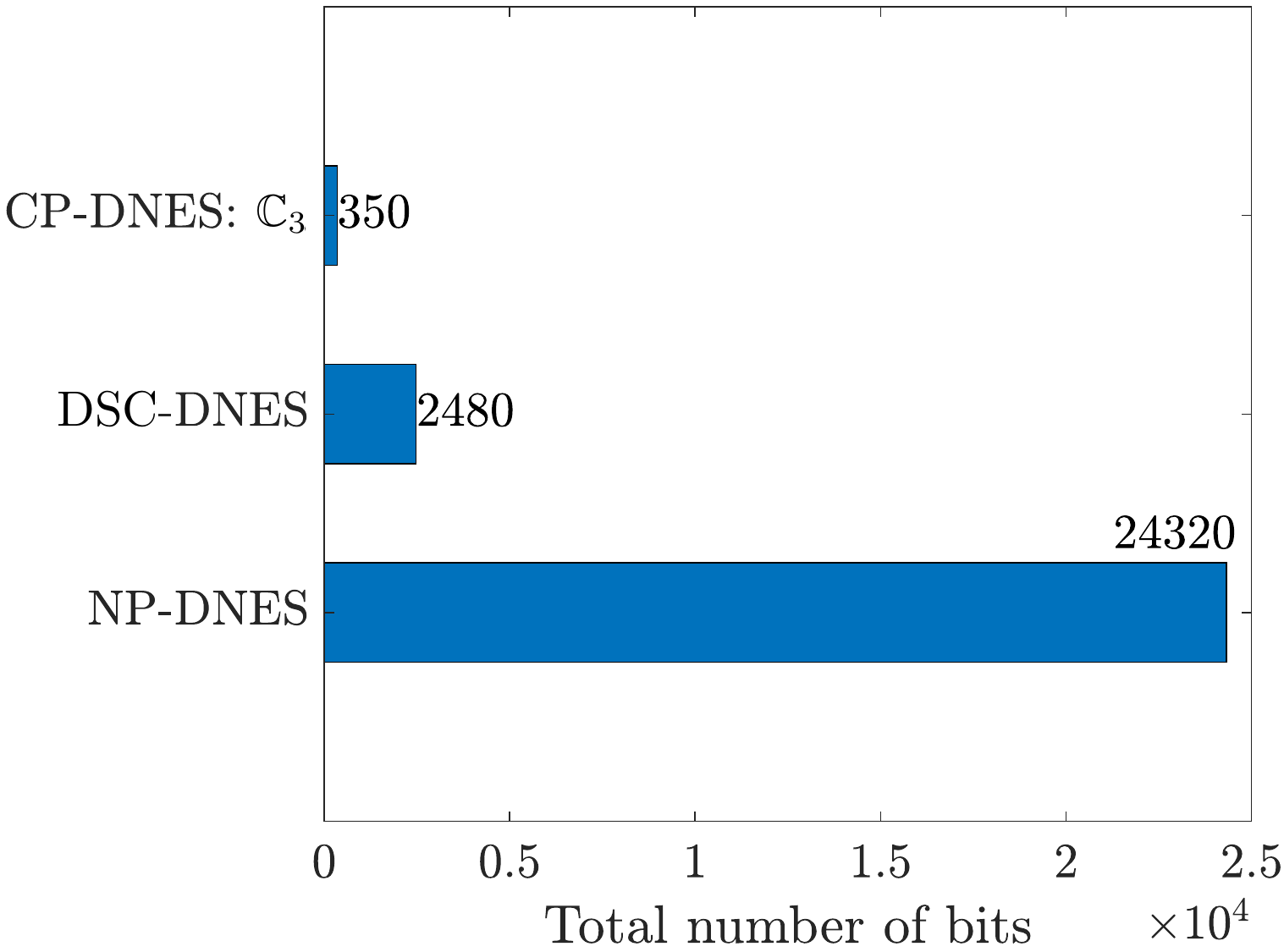}
  \captionsetup{font = small}
		\caption{Total transmitted bits to obtain $\mathbb{E}[\|x_{k}-x^{*} \|] \leq 0.18$.}
		\label{fig: compare_bits}
	\end{subfigure}
\caption{\small{Comparison between NP-DNES~\cite{ye2021differentially}, DSC-DNES~\cite{yi2022communication, liao2023linearly}, and CP-DNES: $\mathbb{C}_{3}$.}}
\label{fig: compare}
\end{figure}

\section{Conclusion and Future Work} \label{sec: conclusions}
This paper investigates privacy preservation in decentralized aggregative games and proposes a distributed NE seeking algorithm robust to aggregative compression effects. 
By incorporating decaying step sizes, we ensure convergence accuracy while leveraging random compression errors to protect shared data and sensitive information.

Potential future research directions include the design of stochastic adaptive compressors to enhance the convergence rate. Additionally, exploiting stochastic event-triggered mechanisms can further reduce communication costs and introduce randomness to bolster privacy protection.

\appendix

\subsection{Useful Lemmas}
The following results are used in the proofs.
%
\begin{lem} \label{lem: pre_con}
~\cite{robbins1971convergence}
	Let $\{ u_{k}\}$, $\{ v_{k}\}$, $\{ w_{k} \}$ and $\{ z_{k} \}$ be the nonnegative sequences of random variables. If they satisfy
	\begin{align*}
		\mathbb{E}[ u_{k+1} ] \leq  (1+ z_{k})u_{k} - v_{k} + w_{k}, \\
		\sum_{k=0}^{\infty} z_{k} < \infty \ a.s., \ \text{and} \  \sum_{k=0}^{\infty} w_{k} < \infty \ a.s.,
	\end{align*}
	then $u_{k}$ converges almost surely (a.s.), and $\sum_{k=0}^{\infty} v_{k} < \infty$ a.s..
\end{lem}

\begin{lem} \label{lem: supermartingals}
~\cite{robbins1971convergence}
	Let $\{ u_{k} \}$ be a non-negative sequence satisfying the following relationship for all $k \geq 0$:
	\begin{equation}
		u_{k+1} \leq (1+ \gamma_{k} ) u_{k} + v_{k},
	\end{equation} 
	where sequence $\gamma_{k} \geq 0$ and $v_{k} \geq 0$ satisfy $\sum_{k=0}^{\infty} \gamma_{k} < \infty$ and $\sum_{k=0}^{\infty} v_{k} < \infty$, respectively. Then the sequence $\{ u_{k} \}$ will converge to a finite value $u>0$.
\end{lem}

\begin{lem}\label{lem: rate}
\cite{kar2011convergence} Let $\{ u_{k} \}$ be a non-negative sequence satisfying the following relationship for all $k \geq 0$:
\begin{equation}
	u_{k + 1} \leq (1- \gamma_{1,k} ) u_{k} + \gamma_{2,k}
\end{equation}
where $\gamma_{1,k} \geq 0$ and $\gamma_{2,k} \geq 0$ satisfying 
\begin{equation*}
	\frac{a_{1}}{(a_{2}k+1)^{b_{1}}} \leq \gamma_{1,k} \leq 1, \
	\frac{a_{3}}{(a_{2}k+1)^{b_{2}}} \leq \gamma_{2,k} \leq 1,
\end{equation*}
for some $c_{1}, c_{2}, c_{3} > 0$, $0 \leq b_{1} < 1$, and $b_{1} < b_{2}$. Then for all $0 \leq b_{0} < b_{2} - b_{1}$, we have $\lim_{k \to \infty}(k+1)^{b_{0}} u_{k} = 0$.
\end{lem}

\subsection{Proof of Lemma 1} \label{app: lem1}

Denote $\bar{\mathbf{y}} = \frac{1}{N} \mathbf{1}^{\top} \mathbf{y}$ and $\Theta = I - \frac{1}{N} \mathbf{1} \mathbf{1}^{\top}$, then we have $ \left\| \Theta \right\| \leq 1$ and
\begin{align} \label{eq: y-x}
	 \left\| \mathbf{y}_{k+1} - \mathbf{1}\bar{\mathbf{x}}_{k+1} \right\|^{2} 
	=& \left\| \Theta \mathbf{y}_{k+1} + \mathbf{1}\bar{\mathbf{y}}_{k+1} - \mathbf{1}\bar{\mathbf{x}}_{k+1} \right\|^{2} \nonumber \\
	\leq  & \left\| \mathbf{y}_{k+1} \right\|^{2} ,
\end{align} 
where the inequality follows from (8) in the main article. 
Therefore, we can complete the proof by showing $ \sum_{k=0}^{\infty} \beta_{k} \mathbb{E}\left[\| \mathbf{y}_{k} \|^{2} \right] < \infty$. 
Since $\sum_{k=1}^{\infty} \beta_{k}^{2} < \infty$, there obviously exists $k_{1}$ such that $\beta_{k} < 1/\lambda_{2}$, $\forall k > k_{1}$. We know that $ \sum_{k=0}^{k_{1}} \beta_{k} \mathbb{E}\left[\| \mathbf{y}_{k} \|^{2} \right]$ is always bounded. Therefore, we only need to focus on proving $ \sum_{k=k_{1}+1}^{\infty} \beta_{k} \mathbb{E}\left[\| \mathbf{y}_{k} \|^{2} \right] < \infty$

Based on (7b) in the main article, we obtain
\begin{align}
	&\mathbb{E}\left[ \| \mathbf{y}_{k+1} \|^{2} \right] \nonumber \\
	 = & \mathbb{E} \left[ \| A_{k} \mathbf{y}_{k} + \mathbf{x}_{k+1} - \mathbf{x}_{k} \|^{2} \right] + \beta_{k}^{2} \mathbb{E}\left[ \| L \mathbf{e}_{k} \|^{2} \right] \nonumber \\
	 & - \mathbb{E} \left[ \beta_{k} (A_{k} \mathbf{y}_{k} + \mathbf{x}_{k+1} - \mathbf{x}_{k})^{\top}L  \mathbf{e}_{k} \right ] \nonumber \\
	 \leq & \mathbb{E}\left[ (\left\| A_{k} \mathbf{y}_{k} \right\| + \left\| \mathbf{x}_{k+1} - \mathbf{x}_{k} \right\|)^{2} \right] + \beta_{k}^{2} \lambda_{M}^{2}N \sigma^{2} \nonumber \\
	 \leq & (1+\nu_{1})(1-\beta_{k} \lambda_{2} )^{2} \mathbb{E}\left[ \| \mathbf{y}_{k} \|^{2} \right] + (1+\frac{1}{\nu_{1}})C^{2}N \alpha_{k}^{2} \beta_{k}^{2} \nonumber \\
	 & + \lambda_{M}^{2}N \sigma^{2}\beta_{k}^{2},
\end{align}
where the first inequality holds from $\mathbb{E}[\mathbf{e}_{k}] = 0$ and $\mathbb{E}\left[ \| \mathbf{e}_{k} \|^{2} \right] \leq N \sigma^{2}$ in Assumption 6, 
the last inequality uses $(a+b)^{2} \leq (1+\nu_{1})a^{2} + (1+\frac{1}{\nu_{1}})b^{2}$ for any $a,b \in \mathbb{R}$ and $\nu_{1}>0$,
and $\lambda_{2}$ and $\lambda_{M}$ are the second smallest eigenvalue and the eigenvalue with the largest magnitude of $L$~\cite{kar2013distributed}, respectively. 
By setting $\nu_{1} = \beta_{k} \lambda_{2}$, we further get
\begin{align} \label{eq: y_bound_1}
	&\mathbb{E}\left[ \| \mathbf{y}_{k+1} \|^{2} \right] \nonumber \\
	\leq & (1-\beta_{k}^{2} \lambda_{2}^{2})(1-\beta_{k} \lambda_{2}) \mathbb{E}\left[ \| \mathbf{y}_{k} \|^{2} \right] + (\alpha_{k}^{2} \beta_{k}^{2} + \frac{\alpha_{k}^{2} \beta_{k}}{\lambda_{2}})C^{2}N \nonumber \\
	& + \lambda_{M}^{2}N \sigma^{2}\beta_{k}^{2} \nonumber \\
	\leq & (1-\beta_{k} \lambda_{2}) \mathbb{E}\left[ \| \mathbf{y}_{k} \|^{2} \right] + (\alpha_{k}^{2} \beta_{k}^{2} + \frac{\alpha_{k}^{2} \beta_{k}}{\lambda_{2}})C^{2}N + \lambda_{M}^{2}N \sigma^{2}\beta_{k}^{2}.
\end{align}
Since $\beta_{k}\lambda_{2}>0$, we can have the following relationship:
\begin{align} 
	\mathbb{E}\left[ \| \mathbf{y}_{k+1} \|^{2} \right] 
	\leq & (1+\beta_{k}^{2}) \mathbb{E} \left[ \left\| \mathbf{y}_{k} \right\|^{2} \right] - \lambda_{2} \beta_{k} \mathbb{E} \left[ \left\| \mathbf{y}_{k} \right\|^{2} \right] \nonumber \\
	&+ (\alpha_{k}^{2} \beta_{k}^{2} + \frac{\alpha_{k}^{2} \beta_{k}}{\lambda_{2}})C^{2}N + \lambda_{M}^{2}N \sigma^{2}\beta_{k}^{2}.
\end{align}
Since $\sum_{k=k_{1}+1}^{\infty} \alpha_{k}^{2} \beta_{k} < \infty$ and $\sum_{k=k_{1}+1}^{\infty} \beta_{k}^{2} < \infty$, it is guaranteed that $\sum_{k=k_{1}+1}^{\infty} [(\alpha_{k}^{2} \beta_{k}^{2} + \frac{\alpha_{k}^{2} \beta_{k}}{\lambda_{2}})C^{2}N + \lambda_{M}^{2}N \sigma^{2}\beta_{k}^{2}] < \infty $.
Thus, we can conclude that $\sum_{k=k_{1}+1}^{\infty}  \beta_{k} \mathbb{E} \left[ \| \mathbf{y}_{k} \|^{2} \right] < \infty$ from Lemma~\ref{lem: pre_con}.

\subsection{Proof of Theorem 1} \label{app: thm1}


Obviously, there exists $k_{0}$ such that $\beta_{k}<1, \forall k > k_{0}$, to satisfy $\sum_{k=0}^{\infty} \beta_{k}^{2} < \infty$.
According to the dynamics in (7a), it can be verified that the following relationship holds:
\begin{align} \label{eq: x_bounded_1}
	& \mathbb{E} \left[ \| \mathbf{x}_{k+1} - \mathbf{x}^{*}  \|^{2} \right]\nonumber \\
 =& \mathbb{E} \left[ \| \mathbf{P}_{\mathcal{X}}[\mathbf{x}_{k} - \alpha_{k}\beta_{k}G(\mathbf{x}_{k}, \mathbf{y}_{k})] \right. \nonumber \\
  &- \left.  \mathbf{P}_{\mathcal{X}}[\mathbf{x}^{*} - \alpha_{k}\beta_{k}G(\mathbf{x}^{*}, \mathbf{1}\bar{\mathbf{x}}^{*})] \|^{2} \right] \nonumber \\
	\leq & \mathbb{E} \left[ \| \mathbf{x}_{k} - \mathbf{x}^{*} - \alpha_{k} \beta_{k}(G(\mathbf{x}_{k}, \mathbf{y}_{k}) - G(\mathbf{x}^{*}, \mathbf{1}\bar{\mathbf{x}}^{*})) \|^{2} \right]\nonumber \\
	=& \mathbb{E} \left[ \| \mathbf{x}_{k} - \mathbf{x}^{*} \|^{2} \right] + \alpha_{k}^{2} \beta_{k}^{2} \mathbb{E} \left[ \| G(\mathbf{x}_{k}, \mathbf{y}_{k}) - G(\mathbf{x}^{*}, \mathbf{1}\bar{\mathbf{x}}^{*}) \|^{2} \right] \nonumber \\
	&- 2 \alpha_{k} \beta_{k} \mathbb{E} \left[ (\mathbf{x}_{k} - \mathbf{x}^{*})^{\top}(G(\mathbf{x}_{k}, \mathbf{y}_{k}) - G(\mathbf{x}^{*}, \mathbf{1}\bar{\mathbf{x}}^{*})) \right] \nonumber\\
	\leq & \mathbb{E} \left[ \| \mathbf{x}_{k} - \mathbf{x}^{*} \|^{2} \right] + 4C^{2}N \alpha_{k}^{2} \beta_{k}^{2}  \nonumber \\
	&- 2 \alpha_{k} \beta_{k}  \mathbb{E} [ ( \mathbf{x}_{k} - \mathbf{x}^{*})^{\top} (G(\mathbf{x}_{k}, \mathbf{y}_{k}) - G(\mathbf{x}^{*}, \mathbf{1}\bar{\mathbf{x}}^{*})) ] ,
\end{align}
where the first inequality holds from the non-expansive property of the projection operation and the last inequality is followed by Assumption 5.
For the last term of~\eqref{eq: x_bounded_1}, we have
\begin{align} \label{eq: optim_2}
	& - 2 \alpha_{k} \beta_{k} (\mathbf{x}_{k} - \mathbf{x}^{*})^{\top}(G(\mathbf{x}_{k}, \mathbf{y}_{k}) - G(\mathbf{x}^{*}, \mathbf{1}\bar{\mathbf{x}}^{*})) \nonumber \\
	=& - 2 \alpha_{k} \beta_{k} (\mathbf{x}_{k} - \mathbf{x}^{*})^{\top} ( G(\mathbf{x}_{k}, \mathbf{y}_{k}) - G(\mathbf{x}_{k}, \mathbf{1}\bar{\mathbf{x}}) \nonumber \\
	& \quad + G(\mathbf{x}_{k}, \mathbf{1}\bar{\mathbf{x}}) - G(\mathbf{x}^{*}, \mathbf{1}\bar{\mathbf{x}}^{*}) ) \nonumber \\
	=& - 2 \alpha_{k} \beta_{k} (\mathbf{x}_{k} - \mathbf{x}^{*})^{\top} ( G(\mathbf{x}_{k}, \mathbf{y}_{k}) - G(\mathbf{x}_{k}, \mathbf{1}\bar{\mathbf{x}})) \nonumber \\
	& - 2 \alpha_{k} \beta_{k} (\mathbf{x}_{k} - \mathbf{x}^{*})^{\top} (\Phi(\mathbf{x}_{k}) - \Phi(\mathbf{x}^{*}) ) \nonumber \\
	\leq & 2 L_{g} \alpha_{k} \beta_{k} \left\| \mathbf{x}_{k} - \mathbf{x}^{*} \right\| \left\| \mathbf{y}_{k} - \mathbf{1} \bar{\mathbf{x}}_{k}  \right\| - 2m\alpha_{k} \beta_{k} \left\| \mathbf{x}_{k} - \mathbf{x}^{*} \right\|^{2} \\
	\leq & \frac{1}{\nu_{2}} L_{g}^{2} \alpha_{k}^{2} \beta_{k} \left\| \mathbf{x}_{k} - \mathbf{x}^{*} \right\|^{2} + \nu_{2} \beta_{k}  \left\| \mathbf{y}_{k} - \mathbf{1} \bar{\mathbf{x}}_{k}  \right\|^{2} \nonumber ,
\end{align} 


From the first two conditions in (9), we can derive that $\lim_{k \to \infty} \alpha_{k} = 0$. Thus, there exists $k^{\prime} > 0$ such that $\alpha_{k^{\prime}} <1$ and we can only focus on the the sequence $\{\mathbf{x}_{k}\}_{k \geq k^{\prime}}$.
By combining~\eqref{eq: x_bounded_1} and~\eqref{eq: optim_2}
, we obtain the following expression:
\begin{align}
	&  \| \mathbf{x}_{k+1} - \mathbf{x}^{*}  \|^{2}  \nonumber \\
	\leq & \| \mathbf{x}_{k} - \mathbf{x}^{*} \|^{2} + 4C^{2}N \alpha_{k}^{2} \beta_{k}^{2} - 2m\alpha_{k} \beta_{k} \| \mathbf{x}_{k} - \mathbf{x}^{*} \|^{2}
\nonumber \\
	& + 2 \bar{L} \alpha_{k} \beta_{k} \| \mathbf{x}_{k} - \mathbf{x}^{*} \| \| \mathbf{y}_{k} - \mathbf{1} \bar{\mathbf{x}}_{k} \|  \nonumber \\
	\leq & \| \mathbf{x}_{k} - \mathbf{x}^{*} \|^{2} + 4C^{2}N \alpha_{k}^{2} \beta_{k}^{2} - 2m\alpha_{k} \beta_{k} \| \mathbf{x}_{k} - \mathbf{x}^{*} \|^{2}
\nonumber \\
	&+ \frac{1}{\nu} \bar{L}^{2} \alpha_{k}^{2} \beta_{k}^{2} \left\| \mathbf{x}_{k} - \mathbf{x}^{*} \right\|^{2} + \nu   \left\| \mathbf{y}_{k} - \mathbf{1} \bar{\mathbf{x}}_{k}  \right\|^{2}, \nonumber
\end{align}
where the last inequality holds from Young's inequality.
By letting $\nu = \frac{{L}_{g}^{2} \alpha_{k} \beta_{k} }{m}$, we further get
\begin{align} \label{eq: x_2}
	& \left\| \mathbf{x}_{k+1} - \mathbf{x}^{*}  \right\|^{2} \nonumber \\
	\leq & \left\| \mathbf{x}_{k} - \mathbf{x}^{*} \right\|^{2} + 4C^{2}N \alpha_{k}^{2} \beta_{k}^{2} - m\alpha_{k} \beta_{k} \left\| \mathbf{x}_{k} - \mathbf{x}^{*} \right\|^{2} \nonumber \\
	& + \frac{\bar{L}^{2} \alpha_{k} \beta_{k} }{m} \left\| \mathbf{y}_{k} - \mathbf{1} \bar{\mathbf{x}}_{k}  \right\|^{2}.
\end{align}
Taking expectation and summing both side of~\eqref{eq: x_2} from $k=k^{\prime}$ to $T$, we have the following relationship:
\begin{align} \label{eq: final}
	& \sum_{k=k^{\prime}}^{T} m \alpha_{k} \beta_{k} \mathbb{E} \left[ \| \mathbf{x}_{k} - \mathbf{x}^{*} \|^{2} \right] \nonumber \\
	\leq & \| \mathbf{x}_{0} - \mathbf{x}^{*} \|^{2} - \mathbb{E} \left[ \| \mathbf{x}_{T+1} - \mathbf{x}^{*} \|^{2} \right] + 4C^{2}N \sum_{k=k^{\prime}}^{T} \alpha_{k}^{2} \beta_{k}^{2} \nonumber \\
	& + \frac{\bar{L}^{2}}{m} \sum_{k=k^{\prime}}^{T} \alpha_{k}\beta_{k} \mathbb{E} \left[ \| \mathbf{y}_{k} - \mathbf{1} \bar{\mathbf{x}}_{k} \|^{2} \right].
\end{align}
When $T \to \infty$, $4C^{2}N \sum_{k=k^{\prime}}^{T} \alpha_{k}^{2} \beta_{k}^{2}$ is bounded and $\left\| \mathbf{x}_{T+1} - \mathbf{x}^{*} \right\|^{2}$ is bounded due to the bounded constraint.
Moreover, $ \sum_{k=k^{\prime}}^{\infty} \alpha_{k}\beta_{k} \mathbb{E} \left[ \| \mathbf{y}_{k} - \mathbf{1}\bar{\mathbf{x}}_{k} \|^{2}  \right]
	\leq  \sum_{k=k^{\prime}}^{\infty} \beta_{k} \mathbb{E} \left[ \| \mathbf{y}_{k} - \mathbf{1}\bar{\mathbf{x}}_{k} \|^{2}  \right] < \infty $ by Lemma 1. 
	Therefore, the right-hand side of~\eqref{eq: final} is always bounded when $T \to \infty$. 
With $\sum_{k=k^{\prime}}^{\infty} \alpha_{k} \beta_{k} =\infty $, we can conclude that $ \mathbb{E} \left[ \|\mathbf{x}_{k} -  \mathbf{x}^{*}\|^{2} \right]$ converges to zero.

\subsection{Proof of Corollary 1} \label{app: cor1}

If $\omega_{1}$ and $\omega_{2}$ satisfy the conditions in Corollary 1, then $\alpha_{k}$ and $\beta_{k}$ will satisfy the conditions in Theorem 1 to ensure that $\mathbb{E} \left[ \| \mathbf{x}_{k} - \mathbf{x}^{*} \|^{2} \right]$ converges to zero.
According to~\eqref{eq: y_bound_1}, we have the following relationship when $k$ is large enough:
\begin{equation*}
	\mathbb{E}\left[ \| \mathbf{y}_{k+1} \|^{2} \right] \leq (1-\beta_{k} \lambda_{2})\mathbb{E}\left[ \| \mathbf{y}_{k} \|^{2} \right] + \eta_{k},
\end{equation*}
where $\eta_{k} = \frac{c_{4}}{(c_{2}k+1)^{\omega_{3}}}$ for some $c_{4} > 0$ and $\omega_{3} = \min\{ 2\omega_{1} + \omega_{2}, 2 \omega_{2} \}$. Therefore, we have 
\begin{equation} \label{eq: y_rate}
	\lim_{k \to \infty}(k+1)^{\omega_{4}} \mathbb{E}\left[ \| \mathbf{y}_{k} \|^{2} \right]  = 0
\end{equation}
based on Lemma~\ref{lem: rate}, where $0 \leq \omega_{4} < \omega_{3}-\omega_{2}$.

Based on~\eqref{eq: final} and~\eqref{eq: y-x}, we have
\begin{align} \label{eq: lhs}
	& \frac{\sum_{k=0}^{T} m \alpha_{k} \beta_{k} \mathbb{E} \left[ \| \mathbf{x}_{k} - \mathbf{x}^{*} \|^{2} \right]}{\sum_{k=0}^{T} \alpha_{k} \beta_{k}}  \\
	\leq & \frac{\| \mathbf{x}_{0} - \mathbf{x}^{*} \|^{2} - \mathbb{E} \left[ \| \mathbf{x}_{T+1} - \mathbf{x}^{*} \|^{2} \right]}{\sum_{k=0}^{T} \alpha_{k} \beta_{k}} 
	+ \frac{4C^{2}N \sum_{k=0}^{T} \alpha_{k}^{2} \beta_{k}^{2}}{\sum_{k=0}^{T} \alpha_{k} \beta_{k}} \nonumber \\
	& \label{eq: final_rate} + \frac{\frac{\bar{L}^{2}}{m} \sum_{k=0}^{T} \alpha_{k}\beta_{k} \mathbb{E} \left[ \| \mathbf{y}_{k} \|^{2} \right]}{\sum_{k=0}^{T} \alpha_{k} \beta_{k}}.
\end{align}
Equation~\eqref{eq: y_rate} indicates that $\mathbb{E}\left[ \| \mathbf{y}_{k} \|^{2} \right]$ is in the same order of $\alpha_{k}^{2}$ or $\beta_{k}$, and thus, the third term of~\eqref{eq: final_rate} converges to zero with a rate $O\left( \frac{1}{(T+1)^{\omega}} \right)$, where $\omega = \min\{ 2 \omega_{1}, \omega_{2} \}$. 
Moreover, the second term of~\eqref{eq: final_rate} converges to zero with a rate $O\left( \frac{1}{(T+1)^{\omega_{1} + \omega_{2}}} \right)$. Since $\omega_{1} + \omega_{2} > \omega$,~\eqref{eq: lhs} will decay with a rate $\omega$. 

\subsection{Proof of Theorem 2} \label{app: thm2}

It can be easily inferred that $\mathbb{E}[\mathcal{C}(x)] = x$ and $\mathbb{E}\left[ \| \mathcal{C}(x) - x \|^{2} \right] \leq \frac{\theta^{2}}{4}$, fulfilling the requirements of Assumption 6. Furthermore, $\sum_{k=0}^{\infty} \frac{c_{4}}{c_{5}k+1} = \infty$ satisfies the first condition in (9).
Thus, this stochastic compressor enables convergence accuracy when the step sizes satisfy other conditions in (9).

From Algorithm 1, it can be seen that given initial state $\{ \mathbf{x}_{0}, \mathbf{y}_{0} \}$, the network topology $W$ and the function set $\mathcal{F}$,
 the observation sequence $\{ \mathcal{O}_{k} \}_{k\geq0}$ is uniquely determined by the compression scheme.
For any pair of adjacent objective function sets $\{ \mathcal{F}\}$ and $\{\mathcal{F}^{\prime}\}$, the eavesdropper is assumed to know the initial states of the algorithm. Thus, $\mathbf{x}_{0} = \mathbf{x}_{0}^{\prime}$ and $\mathbf{y}_{0} = \mathbf{y}_{0}^{\prime}$ based on the same observation. 
Furthermore, the two function sets generate the same outputs, i.e., $\mathcal{C}(y_{i,k})$ and $\mathcal{C}(y_{i,k}^{\prime})$ for all $i \in \mathcal{N}$. The compression errors are independently and identically distributed. Similarly to the nosie-based privacy analysis~\cite{ye2021differentially}, we can conclude that $x_{i,k} = x_{i,k}^{\prime}$ and $y_{i,k} = y_{i,k}^{\prime}$ for $i \neq i_{0}$ and for all non-negative $k$.
  According to (6), we have the following relation for $i_{0}$:
 \begin{align*}
 	y_{i_{0},k+1} =& y_{i_{0},k} + \beta_{k} \sum_{j \in \mathcal{N}_{i}} w_{i,j}(\mathcal{C}(y_{j,k}) - \mathcal{C}(y_{i_{0},k})) - \alpha_{k} \beta_{k} g_{i_{0},k}, \\
 	y_{i_{0},k+1}^{\prime} =& y_{i_{0},k}^{\prime} + \beta_{k} \sum_{j \in \mathcal{N}_{i}} w_{i,j}(\mathcal{C}(y_{j,k}^{\prime}) - \mathcal{C}(y_{i_{0},k}^{\prime})) - \alpha_{k} \beta_{k} g_{i_{0},k}^{\prime}.
 \end{align*}
 Therefore, we have 
 \begin{equation*}
 	\Delta y_{i_{0},k+1} = y_{i_{0},k+1} - y_{i_{0},k+1}^{\prime} = \Delta y_{i_{0},k} - \alpha_{k} \beta_{k}\Delta g_{i_{0},k},
 \end{equation*}
 where $\Delta g_{i_{0},k} = g_{i_{0},k} - g_{i_{0},k}^{\prime}$. Since $\Delta y_{i_{0},0} = 0$, there is
 \begin{align} \label{eq: delta_y}
 	\| \Delta y_{i_{0},k}\| 
 	\leq & \sum_{s=0}^{k-1} \alpha_{s} \beta_{s} \| \Delta g_{i_{0}, s} \| \nonumber \\
 	\leq & 2C \sum_{s=0}^{k-1} \alpha_{s} \beta_{s} \nonumber \\
 	\leq & \frac{2Cc_{4}}{c_{5}} \ln(c_{5}k+1).
 \end{align}
 
Without generality, suppose the attacker's observation at $k$ for $y_{i(j)}, y_{i(j)}^{\prime}$ is $l\theta$.
Similar to the proof of Theorem 3 in Wang and Ba{\c{s}}ar~\cite{wang2022quantization}, we can obtain that $\delta_{k} = \mathbb{P}[\mathcal{C}(y_{i(j)}) = l \theta | y_{i}] - \mathbb{P}[\mathcal{C}(y_{i(j)}^{\prime}) = l \theta | y_{i}^{\prime}] \leq  \frac{\| \Delta y_{i_{0},k} \|_{1}}{\theta}$ depends on $\| \Delta y_{i_{0},k} \|_{1}$. 
	Due to the same observation, there is $| \Delta y_{i_{0}(j),k} | \leq 2 \theta$ from Fig. 1 and $ \| \Delta y_{i_{0},k} \|_{1}  = \sum_{j=1}^{n} | \Delta y_{i_{0}(j),k} | \leq 2n\theta$. 
	Additionally, according to~\eqref{eq: delta_y}, we have $\| \Delta y_{i_{0},k} \|_{1} \leq \sqrt{n} \| \Delta y_{i_{0},k} \| \leq \frac{2Cc_{4}\sqrt{n}}{c_{5}} \ln(c_{5}k+1)$. 
Moreover, it should be noted that in DP, $\delta_{k}$ should be a small parameter in $(0,1)$.
	Hence, we derive the expression of $\delta_{k}$ shown in (12).
	\begin{equation*}
		\delta_{k} \leq \frac{\| \Delta y_{i_{0},k} \|_{1}}{\theta} \leq \min \left\{ 1, \frac{2Cc_{4}\sqrt{n}}{c_{5}\theta} \ln(c_{5}k+1) \right \}.
	\end{equation*}

\bibliographystyle{unsrt}
\bibliography{HW.bib}

\end{document}